\documentclass[11pt]{article}
\usepackage{authblk}

\usepackage[latin1]{inputenc}
\usepackage[T1]{fontenc}
\usepackage{amsmath}
\usepackage{amsfonts}
\usepackage{amssymb}

\usepackage[normalem]{ulem} 

\usepackage{amsmath,xcolor,ulem}
\usepackage{amsfonts}
\usepackage{amssymb,todonotes}
\usepackage{amsthm}
\usepackage{graphicx}
\usepackage{marginnote}

\usepackage[colorlinks=true, allcolors=blue]{hyperref}

\usepackage[top=2.8cm,bottom=2.8cm,left=2.5cm,right=2.5cm]{geometry}

\usepackage{enumitem} 

\newcommand{\bbr}{\mathbb{R}}

\newcommand{\bbc}{\mathbb{C}}

\newcommand{\iu}{\mathrm{i}} 
\DeclareMathOperator{\Imag}{Im}
\DeclareMathOperator{\Real}{Re}

\newcommand{\norm}[1]{\left\lVert#1\right\rVert}
\newcommand{\modulus}[1]{\left\lvert#1\right\rvert}
\newcommand{\inner}[2]{\left( #1 \middle\vert #2 \right)}

\newtheorem{satz}{Satz}[section]

\newtheorem{proposition}[satz]{Proposition}
\newtheorem{corollary}[satz]{Corollary}

\theoremstyle{definition}

\newtheorem{definition}[satz]{Definition}

\newtheorem{remark}[satz]{Remark}

\newtheorem{example}[satz]{Example}

\title{Square Root Operators and the Well-Posedness of Pseudodifferential Parabolic Models of Wave Phenomena}

\author{Matthias Ehrhardt\footnote{Corresponding author, \href{mailto:ehrhardt@uni-wuppertal.de}{ehrhardt@uni-wuppertal.de}},
Jochen Gl\"uck\footnote{ \href{mailto:glueck@uni-wuppertal.de}{glueck@uni-wuppertal.de}} ,
Pavel Petrov\footnote{ \href{pavel.petrov@impa.br}{pavel.petrov@impa.br}} ,
Stefan Tappe\footnote{ \href{mailto:tappe@uni-wuppertal.de}{tappe@uni-wuppertal.de}} 
}

\affil{Applied and Computational Mathematics, University of Wuppertal, Germany}
\affil{Instituto de Matematica Pura e Aplicada, Rio de Janeiro, Brazil}
\affil{Functional Analysis, University of Wuppertal, Germany}
\affil{Stochastics, University of Wuppertal, Germany}

\begin{document}
\maketitle

\begin{tikzpicture}[remember picture,overlay]
	\node[anchor=north east,inner sep=20pt] at (current page.north east)
	{\includegraphics[scale=0.2]{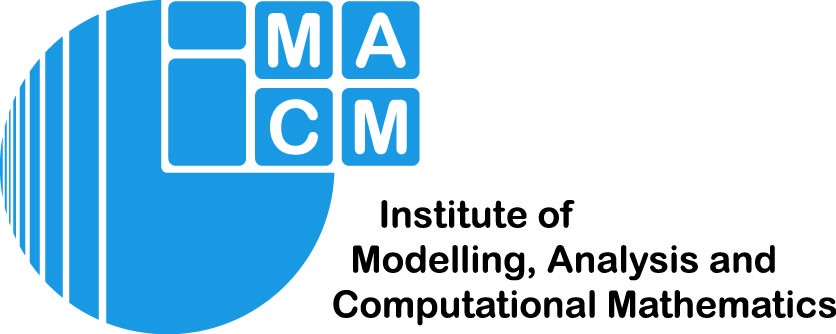}};
\end{tikzpicture}

\begin{abstract}
Pseudodifferential parabolic equations with an operator square root arise in wave propagation problems as a one-way counterpart of the Helmholtz equation. 
The expression under the square root usually involves a differential operator and a known function. 
We discuss a rigorous definition of such operator square roots and show well-posedness of the pseudodifferential parabolic equation by using the theory of strongly continuous semigroups. 
This provides a justification for a family of widely-used numerical methods for wavefield simulations in various areas of physics.   
\end{abstract}

\begin{minipage}{0.9\linewidth}

\medskip

\end{minipage}

\section{Introduction}

A large family of one-way propagation equations used in the numerical modeling of wave phenomena is known as the \textit{parabolic wave equations} (PWEs). 
They have their origin in the work of Leontovich and Fock  \cite{leontovich}, in which a model for the simulation of radio waves was proposed. 
Since then, a wide variety of PWEs have been developed to solve practical problems in seismics, acoustics, optics, and electrical engineering \cite{claerbout1985,jensen2011,LuReview,lytaev2022rational} (sometimes referred to as beam propagation methods). 
In this approach, boundary value problems for the Helmholtz equation in a waveguide are replaced by Cauchy problems for PWEs, which are more convenient to handle numerically using marching solution techniques. 
Another reason for the success of this approach is the ability to cancel out certain oscillatory terms,  thus removing a wavelength resolution limitation associated with the step size.

Currently, the standard approach to deriving PWEs is based on a formal factorization of the Helmholtz operator. 
Such factorization leads to evolutionary equations involving the square root of a differential operator \cite{fishman1984}, sometimes called \textit{pseudodifferential parabolic equations} (PDPEs) (they also appear in the literature under a variety of names, e.g., very-wide-angle parabolic equations). 
Many modern wave propagation techniques are designed to solve PDPEs directly \cite{LuReview,petrov2020,he2024} (rather than first rewriting them, e.g., using some approximation of the square root operator). 
An important example is a powerful method called \textit{Split-Step Pad\'e} (SSP) \cite{collins1994,petrov2024generalization}, which was a breakthrough in both accuracy and performance in underwater acoustics.

Despite the existence of a large number 
of works on PDPEs and their practical importance, to the best of our knowledge the questions of uniqueness, existence and well-posedness for such equations have not been addressed in the literature. 
In fact, most of the research on this topic does not even rigorously define the square root operator \cite{lytaev2022rational,he2024,petrov2020}. 
This letter aims to fill this gap. 
We provide a derivation of the most common PDPE form and rigorously define the square root operator in this equation.
We then prove the uniqueness and existence of the PDPE solution using the semigroup property of the latter operator. 
Our discussion includes 
the piecewise constant dependence of the propagation medium on the range (i.e., on the evolutionary variable).

\section{The Pseudo-Differential Parabolic Equation}\label{sec:pdpe}

Consider the \textit{two-dimensional Helmholtz equation} in Cartesian coordinates $(x,y)$,
\begin{equation}
    \label{2DHE}
    \frac{\partial^2 u}{\partial x^2} + \frac{\partial^2 u}{\partial y^2} + k^2(x,y)\,u = 0\,,
\end{equation}
where $u = u(x,y)$ is the unknown function and $k = k(x,y)$ is a coefficient called medium wavenumber. In practice, $k(x,y)$ is usually a complex quantity with both real and imaginary parts being positive and bounded both from above and from below (the imaginary part represents wave absorption, and it is usually much smaller than the real part).

Assuming that the coefficient dose not depend on $x$ (this is a preferred direction of propagation called waveguide axis along which the medium parameters change very slowly \cite{leontovich,jensen2011,petrov2020}), we can factorize the operator in Eq.~\eqref{2DHE} to get
\begin{equation}\label{2DHE2}
    \left(\frac{\partial }{\partial x} + \iu\sqrt{\frac{\partial^2 }{\partial y^2}+k^2(y)}\right)\left(\frac{\partial }{\partial x} - \iu\sqrt{\frac{\partial^2 }{\partial y^2}+k^2(y)}\right)u = 0\,,
\end{equation}
where the factors on the left-hand side correspond to leftward and rightward one-way propagation of the wave, respectively. 
Without loss of generality, hereafter we consider the latter case and rewrite the one-way counterpart of Eq.~\eqref{2DHE2} as the PDPE
\begin{equation}\label{PDPE}
   \frac{\partial u}{\partial x} = \iu \sqrt{A}\, u
\end{equation}
for $x>0,\;-\infty<y<\infty$, where $A = \frac{\partial^2}{\partial y^2} +  k^2$ is a differential operator acting on the functions of the coordinate $y$ representing a direction transverse to the waveguide axis.
Eq.~\eqref{2DHE} arises in various physical settings. 
For example, in underwater acoustics, $u(x,y)$ can describe the acoustic pressure field in a vertical plane \cite{collins1994,jensen2011}, where $x$ denotes the range and $y$ the depth (in this case,  Eq.~\eqref{2DHE} is usually supplemented by boundary conditions, representing the sea surface and the bottom, e.g., $u(x,0)=0$ and $u(x,H)=0$). 
Eq.~\eqref{2DHE} on the entire half-space $x\geq 0$ can be considered as a one-way counterpart of the horizontal refraction equation for one mode amplitude in the adiabatic approximation \cite{jensen2011,petrov2020}. 

The derivation of Eq.~\eqref{2DHE} is somewhat heuristic. 
However, once it is obtained it is desirable to establish that, for an initial condition $u(0,y) = u_0(y)$ at $x=0$, the Cauchy problem in Eq.~\eqref{PDPE} is well-posed in an appropriate function space. 
Well-posedness means existence and uniqueness of the solution together with continuous dependence on the initial value (see e.g.\ \cite[Section~II.6]{Engel-Nagel} for a more thorough discussion). 
For Eq.~\eqref{PDPE} well-posedness is equivalent to the property that $\iu\sqrt{A}$ generates a strongly continuous semigroup $(T_x)_{x \geq 0}$ \cite[Theorem~II.6.7]{Engel-Nagel}. 
In this case, for each initial value $u_0 = u_0(y)$ the unique solution of Eq.~\eqref{PDPE} is given by $u(x,y) = T_x u_0(y)$.

%

\section{Definition and generator properties of $\iu\sqrt{A}$}\label{sec:sqroots}

%

In this section we discuss criteria for the existence and uniqueness of square roots of unbounded operators on a Hilbert space. 
This will give a precise meaning to the expression $\sqrt{A}$ in the PDPE~\eqref{PDPE}. 
Let $H$ be a complex Hilbert space. 
We use the convention that the inner product is anti-linear in the first argument and linear in the second. 
Let $A\colon H \supseteq D(A) \to A$ be a closed linear operator. 
We denote the \textit{spectrum} of $A$ by $\sigma(A)$; 
its complement $\rho(A) := \bbc \setminus \sigma(A)$ is called the \textit{resolvent set} of $A$.
We will first discuss uniqueness and then existence.
Even for complex numbers (rather than operators) the complex square root is only unique if one imposes an additional assumption on the argument of the root. 
Similarly, the spectral conditions in the following proposition ensure the uniqueness of a square root operator, provided that it exists. 

\begin{proposition}[Uniqueness of square roots]
    \label{prop:square-root-unique}
    Let $H$ be a complex Hilbert space and let $A\colon H \supseteq D(A) \to H$ be a closed linear operator 
    such that $(-\infty,0] \subseteq \rho(A)$. 
    There exists at most one closed linear operator $B\colon H \supseteq D(B) \to H$ with the following properties: 
    \begin{enumerate}[label=\upshape(\alph*)]
        \item\label{prop:square-root:itm:root}
        $B^2 = A$.
        
        \item\label{prop:square-root:itm:spec}
        The spectrum $\sigma(B)$ is contained in the right half plane $\{\lambda \in \bbc \mid \Real \lambda \ge 0\}$.
    \end{enumerate}
\end{proposition}

\begin{proof}
    Consider two closed linear operators $B\colon H\supseteq D(B) \to H$ and $\tilde{B}\colon H\supseteq D(\tilde{B}) \to H$ which satisfy the properties~\ref{prop:square-root:itm:root} and~\ref{prop:square-root:itm:spec}. 
    It follows from the spectral mapping theorem for polynomials 
    \cite[Proposition~A.6.2]{Haase}
    that the imaginary axis does not intersect $\sigma(B)$ nor $\sigma(\tilde{B})$. 
    In particular, $B^{-1}$ and $\tilde{B}^{-1}$ are bounded linear operators on $H$ whose square is equal to $A^{-1}$ and whose spectra are also contained in the open right half plane $\bbc_+ := \{z \in \bbc \mid \Real(z) > 0 \}$. 
    Now consider the holomorphic functions
    \begin{align*}
        f\colon& \bbc_+ \to \bbc \setminus (-\infty,0], \quad z \mapsto z^2, \\ 
        g\colon& \bbc \setminus (-\infty,0] \to \bbc_+, \quad z \mapsto z^{1/2},
    \end{align*}
    i.e., $g$ is the principal branch of the complex square root. 
    Note that the composition $g \circ f$ is the identity function on the domain $\bbc_+$. 
    
    Since $B^{-1}$ and $A^{-1}$ are bounded operators whose spectra are contained in $\bbc_+$ and $\bbc \setminus (-\infty,0]$, respectively, one can use the Dunford functional calculus to compute $f(B^{-1})$ and $g(A^{-1})$ (see e.g.\ \cite[Section~V.8]{TaylorLay1980} or \cite[Section~VIII.7]{Yosida1980} for details on this functional calculus).
    One has $f(B^{-1}) = \big(B^{-1}\big)^2 = \big(B^2\big)^{-1} = A^{-1}$ (here we used the multiplicativity of the functional calculus, see e.g.\ 
    \cite[Theorem~V.8.1]{TaylorLay1980} or \cite[Theorem in Section~VIII.7]{Yosida1980})
    and
    \begin{equation*}
        B^{-1} = (g \circ f)(B^{-1}) = 
        g\big( f(B^{-1}) \big) = g(A^{-1});
    \end{equation*}
    the first equality uses that the functional calculus is compatible with compositions of functions, see \cite[Corollary~2 in Section~VIII.7, p.\,227]{Yosida1980}. 
    The same reasoning can be applied to $\tilde{B}$ instead of $B$, giving
    $B^{-1} = g(A^{-1}) = \tilde{B}^{-1}$.
    Hence, $B = \tilde{B}$, as claimed.
\end{proof}

%
%
%
%

Proposition~\ref{prop:square-root-unique} justifies the following definition.

\begin{definition}[The square root of an operator]
    Let $H$ be a complex Hilbert space and let $A\colon H \supseteq D(A) \to H$ be a closed linear operator 
    such that $(-\infty,0] \subseteq \rho(A)$. 
    We say that \textit{$A$ has a square root} if there exists a closed linear operator $B\colon H \supseteq D(B) \to H$ which satisfies $B^2 = A$ and $\sigma(B) \subseteq \{\lambda \in \bbc \mid \Real \lambda \ge 0\}$. 
    In this case we call $B$ \textit{the square root of $A$} and denote it by $B =: A^{1/2}$. 
\end{definition}

The existence of such a square root is not guaranteed in general.
In the following we discuss an operator theoretic result which ensures the existence of a square root under assumptions that are well-suited to the PDPE~\eqref{PDPE}. 
We need the following concept:
For a closed linear operator $A\colon H \supseteq D(A) \to H$ on a complex Hilbert space $H$, the \textit{numerical range} of $A$ is defined to be the set
\begin{equation*}
    W(A):= \bigl\{ \inner{f}{Af}  \mid  f \in D(A) \text{ and } \norm{f} = 1 \bigr\}.
\end{equation*}
The set $W(A)$ is always convex \cite[Theorem~V.3.1, p.\,267]{Kato1995}. 
The complement of its closure $\overline{W(A)}$ has either one or two connected components, and if one of them intersects $\rho(A)$, then this entire connected component is contained in $\rho(A)$ \cite[Theorem~V.3.2]{Kato1995}.

\begin{proposition}[Existence of square roots]
    \label{prop:square-root-exist}
    Let $H$ be a complex Hilbert space and let $A\colon H\supseteq D(A)\to H$ be a closed linear operator. 
    Assume that there is a $\delta > 0$ such that
    \begin{equation*}
        \sigma(A) \cup \overline{W(A)}
        \subseteq 
        \bbc_{\Imag \ge \delta} 
        := 
        \{\lambda \in \bbc \mid \Imag \lambda \ge \delta \}.
    \end{equation*}
    Then $A$ has a square root and the spectrum and numerical range of $A^{1/2}$ are located in the first quadrant of $\bbc$ and satisfy $\sigma(A^{1/2}) \subseteq \overline{W(A^{1/2})}$.
\end{proposition}
\begin{proof}
    By assumption, the numerical range of the operator $-\iu A$ is contained in the closed right half plane of $\bbc$, which means in the terminology of \cite[Section~V.3.10, p.\,279]{Kato1995} that $-\iu A$ is \textit{accretive}. 
    It also follows from the assumption that the closed left half plane (and thus, in particular, the open left half plane) is in the resolvent set of $-\iu A$ which means, again in the terminology of \cite[Section~V.3.10, p.\,279]{Kato1995}, that $-\iu A$ is even \textit{$m$-accretive}.
    
    Hence, one can apply \cite[Theorem~V.3.35, p.\,281]{Kato1995} which says that there is a closed linear operator $C$ on $H$ that satisfies $C^2 = -\iu A$ and whose numerical range satisfies
    \begin{equation}
        \label{eq:sector-root}
        W(C)
        \subseteq
        \big\{z \in \bbc \mid  \modulus{\arg(z)} \le \frac{\pi}{4} \big\}.
    \end{equation} 
    Moreover, by the same theorem the operator $C$ is \textit{$m$-sectorial} (see \cite[Section~V.3.10, p.\,280]{Kato1995} for the definition of this notion), so in particular all $\lambda \in \bbc$ with sufficiently negative real part are in the resolvent set of $C$. 
    So the complement of $\overline{W(C)}$ intersects the resolvent set of $C$; 
    moreover, it follows from~\eqref{eq:sector-root} and the convexity of $W(C)$ \cite[Theorem~V.3.1, p.\,267]{Kato1995} that the complement of $\overline{W(C)}$ is connected. 
    The facts mentioned before Proposition~\ref{prop:square-root-exist} thus imply $\sigma(C) \subseteq \overline{W(C)}$.
    We conclude that the operator $B := e^{\iu \frac{\pi}{4}} C$ has all the required properties.
\end{proof}

Note that the reference \cite[Theorem~V.3.35, p.\,281]{Kato1995}, which we used in the proof, does not only give existence but also uniqueness of square roots -- but only among all accretive operators. 
As a consequence of Proposition~\ref{prop:square-root-exist} one gets the following well-posedness result for differential equations.

%

\begin{corollary}[Generation theorem for $\iu$ times a square root]
    \label{cor:square-root-generation}
    Under the assumptions of Proposition~\ref{prop:square-root-exist} the operator $\iu A^{1/2}$ generates a contractive $C_0$-semigroup on $H$.
\end{corollary}

\begin{proof}
    According to Proposition~\ref{prop:square-root-exist} the square root $A^{1/2}$ exists, and the spectrum and numerical range of $\iu A^{1/2}$ are located in the second quadrant of $\bbc$. 
    This means that $\iu A^{1/2}$ is \textit{$m$-dissipative} (which is another term for saying that minus the operator is $m$-accretive), so it generates a contractive $C_0$-semigroup according to the Lumer--Phillips generation theorem \cite[Corollary~II.3.20]{Engel-Nagel}.
\end{proof}

The following example shows why the operator $A = \frac{\partial^2}{\partial y^2} +  k^2$ in the PDPE~\eqref{PDPE} satisfies the assumptions of Proposition~\ref{prop:square-root-exist} and Corollary~\ref{cor:square-root-generation} if the real and imaginary parts of $k$ are positive and bounded away from $0$. 
In this case one can choose $L = \partial^2/\partial y^2$ and $m = k^2$.
\begin{example}
    \label{exa:selfadjoint-plus-mult}
    Let $H = L^2(\Omega)$ for a domain $\Omega \subseteq \bbr^n$, let $L\colon H \supseteq D(L) \to H$ be a self-adjoint linear operator and let $m\colon \Omega \to \bbc$ be a bounded and continuous (or, more generally, bounded and measurable) function that satisfies $\Imag(m(\omega)) \ge \delta$ for a $\delta > 0$ and all $\omega \in \Omega$. 
    Then the operator $A := L+m$ with domain $D(A) := D(L)$ satisfies $\sigma(A) \cup \overline{W(A)} \subseteq \bbc_{\Imag \ge \delta}$, so Proposition~\ref{prop:square-root-exist} and Corollary~\ref{cor:square-root-generation} are applicable to $A$.
\end{example}

\begin{proof}
    Since $L$ is self-adjoint one has $W(L) \subseteq \bbr$, and it is easy to check that $W(m)$ is contained in the closure of the range of $m$. 
    Hence, 
    \begin{align*}
        W(A) 
        \subseteq 
        W(L) + W(m) 
        \subseteq 
        \bbr + \bbc_{\Imag \ge \delta} 
        \subseteq 
        \bbc_{\Imag \ge \delta}
        .
    \end{align*}
    Since $m$ is a bounded perturbation, all numbers with sufficiently negative imaginary part are contained in $\rho(L+m)$. 
    To see this, first note that $\norm{(\lambda-L)^{-1}} \le \frac{1}{\modulus{\Imag \lambda}}$ for all $\lambda \in \bbc \setminus \bbr$ since $L$ is self-adjoint and then conclude that $\lambda - (L+m) =  \big(1 - m(\lambda-L)^{-1} \big)(\lambda-L)$ is invertible for $\norm{m}_\infty < \modulus{\Imag \lambda}$ by using the Neumann series. 
    So the connected component of $\bbc \setminus \overline{W(A)}$ that contains $\bbc \setminus \bbc_{\Imag \ge \delta}$ intersects the resolvent set $\rho(A)$ and hence is contained in $\rho(A)$, as pointed out before Proposition~\ref{prop:square-root-exist}. 
    So $\sigma(A) \subseteq \bbc_{\Imag \ge \delta}$.
\end{proof}

\begin{remark}
    Example~\ref{exa:selfadjoint-plus-mult} gives the well-posedness of the PDPE~\eqref{PDPE} under the assumptions discussed before the example. 
    The argument assumed that $k^2(x,y)$ does not depend on $x$, but it can be directly generalized to the case where $k^2(x,y)$ is piecewise-constant in $x$. 
    More precisely, assume that the interval $x\in [0,L]$ is divided into a set of $N$ subintervals $[x_{j-1},x_{j}]$, $j = 1,\dots,N$, where $x_0=0$, $x_{N}=L$, and $k^2(x,y) = k_j^2(y)$ for all $x\in [x_{j-1},x_{j}]$. 
    Then the Cauchy problem~\eqref{PDPE} can be solved piecewise on the subintervals $[x_{j-1},x_{j}]$.
\end{remark}

\section{Conclusion}

We provided a theoretical foundation 
for the PDPEs theory, which is widely used in the numerical simulation of wave dynamics \cite{jensen2011,LuReview,lytaev2022rational,he2024,collins1994}. 
First, we gave a rigorous definition of the square root operator appearing in such equations. 
Second, we established the well-posedness
%
%
%
of the corresponding Cauchy problem. 
Due to the abstract nature of the proofs, they cover most typical Cauchy problem setups that arise in practice (e.g., it is sufficient that the initial data $u_0$ is square integrable), although the dependence of the problem parameters on the range (i.e., on $x$) is restricted to piecewise constant functions.
 In many physics and engineering problems, this is exactly the way the information about the medium is usually given \cite{jensen2011}.
 On the other hand, it is desirable to establish the same results as above for more general types of range-dependent media (we plan to address this in future work).

\section*{Acknowledgements.}

We are grateful to Markus Haase and Christian Wyss for several helpful discussions.



\begin{thebibliography}{00}

\bibitem{leontovich}
M.~Leontovich, V.~Fock, Solution of the problem of electromagnetic wave
  propa\-gation along the earth`s surface by the method of parabolic equation,
  J. Phys. USSR 10 (1946) 13--23.

\bibitem{claerbout1985}
J.~Claerbout, Fundamentals of Geophysical Data Processing: With Applications to
  Petroleum Prospecting, Blackwell Scientific Publications, 1985.

\bibitem{jensen2011}
F.~B. Jensen, M.~B. Porter, W.~A. Kuperman, H.~Schmidt, Computational Ocean
  Acoustics, 2nd Edition, Springer, 2011.

\bibitem{LuReview}
Y.~Y. Lu, Some techniques for computing wave propagation in optical waveguides,
  Commun. Comput. Phys. 1~(6) (2006) 1056--1075.

\bibitem{lytaev2022rational}
M.~Lytaev, Rational interpolation of the one-way {H}elmholtz propagator, J.
  Comput. Sci. 58 (2022) 101536.

\bibitem{fishman1984}
L.~Fishman, J.~McCoy, Derivation and application of extended parabolic wave
  theories. {I}. {The} factorized {Helmholtz} equation, J. Math. Phys. 25~(2)
  (1984) 285--296.

\bibitem{petrov2020}
P.~S. Petrov, X.~Antoine, Pseudodifferential adiabatic mode parabolic equations
  in curvilinear coordinates and their numerical solution, J. Comput. Phys. 410
  (2020) 109392.

\bibitem{he2024}
T.~He, J.~Liu, S.~Ye, X.~Qing, S.~Mo, A novel model order reduction technique
  for solving horizontal refraction equations in the modeling of
  three-dimensional underwater acoustic propagation, J. Sound Vibr. 591 (2024)
  118617.

\bibitem{collins1994}
M.~D. Collins, Generalization of the split-step {P}ad{\'e} solution, J. Acoust.
  Soc. Amer. 96~(1) (2015) 382--385.

\bibitem{petrov2024generalization}
P.~S. Petrov, M.~Ehrhardt, S.~B. Kozitskiy, A generalization of the split-step
  {P}ad{\'e} method to the case of coupled acoustic modes equation in a 3d
  waveguide, J. Sound Vibr. 577 (2024) 118304.

\bibitem{Engel-Nagel}
K.-J. Engel, R.~Nagel, One-Parameter Semigroups for Linear Evolution Equations,
  Vol. 194 of Graduate Texts in Mathematics, Springer, 2000.

\bibitem{Haase}
M.~Haase, The Functional Calculus for Sectorial Operators, Vol. 169 of Operator
  Theory: Advances and Applications, Springer, 2006.

\bibitem{TaylorLay1980}
A.~Taylor, D.~Lay, Introduction to Functional Analysis, 2nd Edition, {John}
  {Wiley} \& {Sons}, 1980.

\bibitem{Yosida1980}
K.~Yosida, Functional Analysis, 6th Edition, Vol. 123 of Grundlehren Math.
  Wiss., Springer, Cham, 1980.

\bibitem{Kato1995}
T.~Kato, Perturbation Theory for Linear Operators., Vol. 132 of Class. Math.,
  Springer, 1995.


\end{thebibliography}
\end{document}